\documentclass{article}
\usepackage{ijcai24}

\usepackage{times}
\usepackage{soul}
\usepackage{url}
\usepackage[hidelinks]{hyperref}
\usepackage[utf8]{inputenc}
\usepackage[small]{caption}
\usepackage{graphicx}
\usepackage{amsmath}
\usepackage{amsthm}
\usepackage{booktabs}
\usepackage{algorithm}
\usepackage{algorithmic}
\usepackage[switch]{lineno}

\usepackage{latexsym}

\usepackage{amsthm,amssymb} 
\usepackage{enumitem}
\usepackage{centernot}

\newtheorem{theorem}{Theorem}
\newtheorem{definition}{Definition}
\newtheorem{lemma}{Lemma}
\newtheorem{ax}{Axiom}
\newtheorem{axprime}[ax]{Axiom*}

\newenvironment{lproof}[1]{\noindent \emph{Proof of Lemma \phantomsection\ref{#1}.}}{\hfill$\blacksquare$ \\}

\def \R{\mathbb{R}}
\def \s{\operatorname{\succsim}}

\def \w{\omega}

\def \S{\mathcal{S}}
\def \L{\mathcal{L}}
\def \M{\mathcal{M}}
\def \C{\mathcal{C}}
\def \V{\mathcal{V}}
\def \U{\mathcal{U}}
\def \Rn{\mathcal{R}}
\def\A{\mathbb{A}}

\def \F{\mathcal{F}}

\def\d{\mathbf{do}}
\def\at{\mathcal{A}(\S)}
\def\u{\vec{u}\hspace{.1ex}}
\def\s{\succsim}
\def\fixes{\mathrel{\sim\!\succ}}
\newcommand{\pair}[1]{\langle\vec{#1},\vec{#1}^{\prime}\rangle}
\newcommand{\nciteyear}[1]{[\citeyear{#1}]}
\newcommand{\vecc}[1]{\vec{#1}\hspace{.15em}}
\newcommand{\svc}[1]{{\vec{\uppercase{#1}} \leftarrow \vecc{\lowercase{#1}}}}
\newcommand{\ms}[1]{\{\!\{#1\}\!\}}

\def\p{\textup{\textbf{p}}}
\renewcommand{\mu}{\textup{\textbf{u}}}
\newcommand{\commentout}[1]{}

\newcommand{\fullv}[1]{#1}
\newcommand{\shortv}{\commentout}

\title{Subjective Causality}

\author{
Joseph Y. Halpern$^1$
\and
Evan Piermont$^2$
\affiliations
$^1$Computer Science Department, Cornell University, Ithaca, USA.\\
$^2$Department of Economics, Royal Holloway, University of London, UK.
\emails
halpern@cs.cornell.edu,
evan.piermont@rhul.ac.uk
}

\begin{document}

\maketitle

\begin{abstract}
    We show that it is possible
  to understand and identify a decision maker's subjective causal
judgements by observing her preferences over interventions.
Following
Pearl \nciteyear{pearl:2k}, we represent causality using  \emph{causal
models} (also called \emph{structural equations models}), where the world is described by a collection of variables,
related by equations. We show that if a preference relation over interventions satisfies certain axioms (related to standard axioms regarding
counterfactuals), then we can define (i) a causal model, (ii) a probability capturing the decision-maker's uncertainty regarding the
external factors in the world and (iii) a utility on outcomes such that each intervention is associated with an expected utility and
such that intervention $A$ is preferred to $B$ iff the expected
utility of $A$ is greater than that of $B$. In addition, we characterize when
the causal model is unique.  
Thus, our results allow a modeler to test the hypothesis that a
decision maker's preferences are consistent with some causal model and
to identify causal judgements from observed behavior. 
\end{abstract}

\section{Introduction}

Causal judgments play an important role in decision making. When
deciding between actions that intervene directly on some aspect of the
world, one major source of uncertainty is the indirect effect of such
actions via causal interaction. For example, when deciding the
interest rate, the Federal Reserve might consider the possibility that
a change in the interest rate will cause a change in unemployment, and
further that this causal relationship itself might be contingent on
other macroeconomic variables.  

Uncovering and describing the causal relationship between variables is
a task that 
has led to enormous effort across many different disciplines (see,
e.g., 
\cite{angrist2009mostly,cunningham2021causal,HR20a,MW07,parascandola2001causation,plowright2008causal,Pearl09,pearl:2k,SpirtesSG};
this list barely scratches the surface.) 
%
Different decision makers, on account of their private
information and personal experience, might hold different beliefs
about causal relationships. In this paper, we show that it is possible
to understand and identify a decision maker's subjective causal
judgements by observing her preferences over interventions.

A first step to doing this is to decide how to represent causality.
Most recent work has focused on using counterfactuals.  In the
philosophy community, following Stalnaker \nciteyear{Stalnaker68} and
Lewis \nciteyear{Lewis73}, counterfactuals are given semantics using
possible worlds equipped 
with a ``closer than'' relation; a counterfactual such as ``If $\phi$
were true then $\psi$ would be true'' is true at a world $\omega$ if,
$\psi$ is true at the closest world(s) to $\omega$ where $\phi$ is
true.  Pearl \nciteyear{pearl:2k} has championed the use of \emph{causal
models} (also called \emph{structural equations models}, graphical
models where the world is described by a collection of variables,
related by equations.  (These are related to models of causality in
economics that go back to the work of Haavelmo \nciteyear{Haa43} and
Simon \nciteyear{Simon53}.)
The equations model the effect of counterfactual interventions.

We use the latter approach here, although as we we shall show, there
are close relationships with the former approach as well.  Following
Pearl, 
we assume that the world is described by a set of variables. 
It is useful to split them into two sets: the \emph{exogenous} variables,
whose values are determined by factors outside the model, and the
\emph{endogenous} variables, whose values are determined by the
exogenous variables and other endogenous variables.  Which variables
should be taken as exogenous and which should be taken as endogenous
depends on the situation.  For example, if
the Federal Reserve is deciding whether to change the interest rates,
the interest rates should clearly be viewed as endogenous.  But for a
company trying to decide whether to go ahead with a project that
involves borrowing money at the current interest rates, the interest
rates are perhaps best viewed as exogenous.  

A primitive action is an intervention that sets the value of a
particular variable; because the 
values of the exogenous variables are taken as given, we assume that
only on endogenous variables can be intervened on.
Following Pearl \nciteyear{pearl:2k}, we use the $\d$ notation to
denote such actions.
For example $\d[Y\gets y]$ is the primitive action that sets variable
$Y$ to value $y$.
Following Blume, Easley, and Halpern \nciteyear{BEH06} (BEH from now
on), we allow more general actions to be formed from these primitive
actions using {\bf if} \ldots {\bf then} \ldots {\bf else}.  Thus,
non-primitive actions are conditional interventions of the form
$$ \textbf{if } \phi  \textbf{ then } A \textbf{ else } B$$,
where $A$ and $B$ are themselves (possibly primitive) actions and
$\phi$ is a test---a statement regarding the values of variables that
is either true or false. 

As is standard in decision-theoretic analyses, we assume that the
decision maker has a preference relation $\s$ over actions.
\commentout{
In the tradition of \cite{Pearl95}, we model causal relationships
through a set of \emph{structural equations}. Loosely speaking, a
\emph{causal model} is a set of structural equations that determine
the value of each endogenous variable, given the values of other
variables. We focus on \emph{recursive} causal models, precluding
cyclic relationships of dependence between the variables. From the
values of exogenous variables, a recursive causal model completely
(and uniquely) determines the values of all other variables, as well
as the effect of any intervention such as $\d[Y \gets y]$. 
}
We show that if the preference relation satisfies certain axioms (that
can be understood as corresponding to standard axioms regarding
counterfactuals), then we can represent the decision-maker's
preference as the maximization of the expected utility of an actions
relative to a causal model equipped with a probability and utility.
Specifically, given a preference
relation, we can define a causal model, add a probability on
\emph{contexts} (settings of the exogenous variables, which can be
viewed as capturing the decision-maker's uncertainty regarding the
external factors in the world), and a utility on outcomes (which are
 just settings of the variables in the model).  In such a causal
model, we can determine the expected utility of an action that
involves (conditional) interventions.  
We show that the decision maker
prefers action $A$ to $B$ iff the expected utility of $A$ is greater than that of $B$, 
given the causal model, probability, and utility.

Our results allow a modeler to test the hypothesis that a
decision maker's preferences are consistent with some causal model. In
doing so, we provide a definition of \emph{causally sophisticated}
behavior as a benchmark of rationality for decision
making in the presence of interventions.%
\footnote{This is somewhat analogous to \emph{probabilistic
sophistication} in decision making \cite{MS92}.}
When a decision maker is causally
sophisticated---when her actions can be rationalized by some causal
model---our results further determine when the causal model can be
uniquely identified. This result provides a modeler the means to
explain observed economic behavior in terms of causal judgments.

\commentout{
Our main result is an axiomatic characterization of \emph{subjective
causal expected utility}, wherein $\s$ results from 
maximizing the expected utility of an action, relative
to uncertainty regarding the values of exogenous
variables and the structural equations that dictate causal relationships. Specifically, the representation is governed by 
\begin{itemize}
    \item $M$: a model  that dictates the causal equations
    \item $\p$: a probability distribution over values for exogenous variables
    \item $\mu$: a utility function over final values of all variables (i.e., the values that result from the action)
\end{itemize}

Then $(M,\p,\mu)$ represents $\s$ if $A \s B$ if and only if

\begin{equation*}
\sum_{\u \in \C} \mu(\beta^M_A(\u))\p(\u) \geq \sum_{\u \in
  \C} \mu(\beta^M_B(\u))\p(\u). 
\end{equation*}
Where $\C$ is the set of all possible initial values for exogenous variables, and for each $\u \in \C$, $\beta^M_A(\u)$ is the ex-post values of (both exogenous and endogenous) variables that results from taking action $A$, given $\u$, according to the model $M$.\footnote{See Section \ref{sec:rep} for a formal definition of $\beta$.}
}

To the best of our knowledge, we are the first to examine causal decision
making in the context of the structural equations, which is perhaps
now the most common approach to representing causality in the social sciences
and computer science.
Bjorndahl and Halpern \nciteyear{BH21} (BH from now on) addressed
similar questions in 
 the context of the closest-world approach to counterfactuals of Lewis
\nciteyear{Lewis73} and Stalnaker \nciteyear{Stalnaker68}.
Interestingly, our technical results use their results (and 
earlier results of BEH) as the basis for our representation theorem,
based on a connection between the two representations of
counterfactuals due to Halpern \nciteyear{Hal40}.
Our results allow us to relate the two approaches at a
decision-theoretic level.

Other approaches to modeling causality have also been considered in the
literature.  
Schenone \nciteyear{schenone2020causality} and Ellis and Thysenb
\nciteyear{ellis2021subjective} take a statistical 
approach, taking a lack of conditional independence as a definition of
causality; Alexander and Gilboa \nciteyear{alexander2023subjective}
understand causality through 
a reduction in the Kolmogorov complexity. 
As we said, our approach is closer to the approach that is currently
the focus of work in the social sciences and computer science.
Importantly, by identifying the structural equations, we provide a
more detailed insight into the causal mechanisms being considered by
the decision maker.


In our representation (like that of BH), both utility and probability are
defined on valuations of variables. 
So, in contrast to the decision-theoretic
tradition following Savage (and, of particular relevance,
\cite{schenone2020causality}), there is no separation between 
states (on which probability is defined) and outcomes
(on which utilities are defined): our decision maker derives utility directly from the outcome of the intervention.
This permits the application of our model to the many economically
relevant contexts where the effect of the intervention has direct
utilitarian consequences for the decision maker (e.g.,, a government
setting policy, or a firm deciding an investment strategy).

\section{Causal Models}


Let $\U$ and $\V$ denote the set of exogenous and endogenous
variables, respectively.  Let $\Rn: \U \cup \V \to 2^\R$ associate to
each $Y \in \U \cup \V$ a 
set $\Rn(Y)$ of possible values  called its range.  
We assume that the range for each variable is finite.
Let $\S = (\U,\V,\Rn)$ denote a \emph{signature}.

A \emph{causal model} is a pair $M = (\S,\F)$, where $\S$ is a signature
and
$\F$ is a collection of structural equations that determine the values
of endogenous variables based on the values of other variables.  
Formally, $\F = \{F_X\}_{X \in \V}$, where 
$$F_X: \prod_{Y \in \U \cup (\V - \{X\})} \Rn(Y) \to \Rn(X).$$

A model $M$ is \emph{recursive} if there exists a partial order on
$\V$ such that the structural equations for each variable is
independent of the variables lower in the order. 
When we say that $X$ is independent of $Y$, we mean that $F_X$ does not depend on the value of $Y$.
Given a signature $\S$, let $\M(\S)$ denote the set of all models over the signature $\S$ and $\M^{rec}(\S)$ the set of recursive models.

We use the standard vector notation to denote sets of variables or
their values. If $\vec{X} = (X_1, \ldots, X_n)$ is a vector, we write,
in a standard extension of the usual containment relation, $\vec{X}
\subseteq \U \cup \V$ to indicate that each $X_i \in \U \cup
\V$. Similarly, we
write $\vec{x} \in \Rn(\vec{X})$ if $\vec{x} \in \prod_{i \leq n}
\Rn(X_i)$.
For a vector $\vec{Y}$ of variables, $\vec{y} \in \Rn(\vec{Y})$,
and a vector $\vec{X} \subseteq \vec{Y}$, let $\vec{y}|_{\vec{X}}$
denote the restriction of $\vec{y}$ to $\vec{X}$. In particular, for a
single variable $X$, $\vec{y}|_{X}$ denotes the $X^{th}$
component.  

A \emph{context} is a vector $\u$ of values for all the exogenous
variables $\U$.
Let $\C(\S) = \prod_{Y \in \U } \Rn(Y)$ consist of all
contexts for the signature $\S$. Call a pair $(M, \u) \in \M^{rec}(\S)
\times \C(S)$ a \emph{situation}. Each situation has a unique
solution (i.e., a unique value for each variable that simultaneously
satisfies all the structural equations and agrees with the context).  

Given a model $M = (S,\F)$, we can construct a new model (over the
same signature) that represents the counterfactual situation where
some variables are set to specific values. 
If $\vec{Y} \subseteq \V$ and $\vec{y} \in \Rn(\vec{Y})$,
then $M_{\d[\svc{Y}]} = (\S, \F_{\d[\svc{Y}]})$
denotes the model that is identical to $M$ except that the equation for
each variable $X \in \vec{Y}$ in $\F_{\d[\svc{Y}]}$ is replaced
by $X = \vec{y}|_X$.%
\footnote{It is more standard in the literature to write $M_{\svc{Y}}$
rather than $M_{\d[\svc{Y}]}$.  We use the latter notation because it
makes it easier to express some later notions.}
Note that $M_{\d[\svc{Y}]}$ is recursive if $M$ is.

\subsection{Syntax and Semantics}

Fix a signature $\S$. For each $X \in \U \cup \V$ and $x \in
\Rn(Y)$, let $X=x$ denote the atomic proposition that says that the variable
$X$ takes value $x$. Let $\L(\S)$ denote the language constructed
by starting with these propositions and closing off under negation,
conjunction, and disjunction.%
\footnote{It is more standard to have the atomic propositions involve
only endogenous variables, but for our purposes, it is important to
include exogenous variables as well.  We explain why when we discuss
actions below.}

A formula $\phi \in \L(\S)$ is either true or false in a
situation $(M,\u) \in \M^{rec}(\S) \times \C(\S)$. We write
$(M,\vec{u}) \models \psi$  if $\psi$ is true in
the situation $(M,\vec{u})$.  The $\models$ relation is defined inductively.

\begin{itemize}
  \item $(M, \u) \models$ $X=x$ iff $X=x$ in the unique solution to the system
  of equations $\F$, starting with context $\u$.
\item $(M, \u) \models \neg \phi$ iff not $(M, \u) \models \phi$.
\item $(M, \u) \models \phi \land \psi$ iff $(M, \u) \models \phi$ and $(M, \u) \models \psi$.
\item $(M, \u) \models \phi \lor \psi$ iff $(M, \u) \models \phi$ or $(M, \u) \models \psi$.
\end{itemize}

An \emph{atom} is a complete description of the values of variables;
its truth will completely determine truth of all of formulae in
$\L(\S)$. Formally, 
an \emph{atom over $\S$} is a conjunction of the form $\land_{Y \in \U
  \cup \V} Y=y$, where $y$ 
is some value in $\Rn(Y)$.  We use $\vec{a}$ to denote a generic atom.
Let $\at$ denote the set of atoms over $\S$.  Notice that an atom
$\vec{a} \in \at$ determines the truth of all formulas in $\L(\S)$; we
write $\vec{a} \Rightarrow \phi$ if $\phi$ is true in situations satisfying
$\vec{a}$.  
Let $\vec{a}_{M,\vec{u}}$ denote unique atom such that $(M,\vec{u})
\models \vec{a}$.

In a slight abuse of notation, identify each context $\u \in \C$ with the
formula characterizing $\u$, that is, the formula
$\land_{U \in \U} U=\u|_{U}$.

For our later discussion, it is useful to consider an extension
of $\L(\S)$ that includes
formulas of the form $[\vec{Y} \gets \vec{y}]\phi$, where $\phi \in
\L(\S)$, $\vec{Y} \subseteq \V$, and 
$\vec{y} \in \Rn(\vec{Y})$.  
We can view this formula as saying ``after
intervening to set the variables in $\vec{Y}$ to $\vec{y}$, the formula
$\phi$ holds''. 
 Call this extended language $\L^+(\S)$.  We can extend
the semantics that we gave to formulas in 
$\L(\S)$ as follows:
\begin{itemize}
\item  $(M,\vec{u}) \models [\vec{Y} \gets
  \vec{y}]\phi$ iff $(M_{\d[\vec{Y} \gets \vec{y}]}, \vec{u}) \models
  \phi$.
\end{itemize}

In the sequel, we will also make use of another approach to giving
semantics to counterfactuals, due to David Lewis and Robert Stalnaker
\cite{Lewis73,Stalnaker68}.  This approach is based on the idea of ``closest
worlds''; roughly speaking, with this approach, $[\vec{Y} \gets
  \vec{y}]\phi$ is true at a world $\omega$ if $\phi$ is true at all
the worlds closest to $\omega$ where $\vec{Y} = \vec{y}$.  Lewis
formalizes the idea of ``closest world'' 
using a ternary relation $R$ where, for each $\omega \in \Omega$, $R(\omega,
\cdot, \cdot)$ is a partial order on $\Omega.$  
We can
interpret $R(\omega_1,\omega_2, \omega_3)$ as saying that $\omega_2$ is
closer to $\omega_1$ than $\omega_3$ is.  (In this paper, we focus on the case
that, for each world $\omega$, $R(\omega, \cdot, \cdot)$ is a strict
linear order.)

In more detail, a \emph{Lewis-style model} is a tuple $M=(\Omega,R, I)$,
where $R$ is a ternary relation on $\Omega$ as above, and $I$ is an
\emph{interpretation} that determines whether each atomic proposition
is true or false at each state.  Formally, if $AP$ is the set of
atomic propositions (as implicitly determined by a set $\U \cup \V$ of
exogenous and endogenous variables), then $I : \Omega \times AP
\rightarrow \{{\bf true, \, false}\}$.  We can again define a
relation $\models$ by induction, by taking
\begin{itemize}
\item $(M,\omega) \models X=x$ iff $I(\omega,X=x) = {\bf true}$,
\end{itemize}
defining the semantics of negation, conjunction, and disjunction as
above, and for interventions, taking
\begin{itemize}
\item $(M,\omega) \models [\vec{Y} \gets   \vec{y}]\phi$ iff $(M,\omega')
\models \phi$, for all $\omega'$ that are minimal worlds according to
 the order $R(\omega,\cdot,\cdot)$ such that $(M,\omega') \models \vec{Y} =  \vec{y}$; 
 that is, the worlds closest to $\omega$ for which $\vec{Y} =  \vec{y}$ holds.
\end{itemize}
As shown by Halpern \nciteyear{Hal40}, recursive causal models
correspond in a precise sense to a subclass of Lewis-style models; we
return to this point in Section \ref{sec:repPrf}.

\section{Decision Environment}

A \emph{primitive action} over $\S$ has
the form $\d[Y_1  \gets y_1, \ldots, Y_n 
\gets y_n]$, abbreviated as $\d[\vec{Y} \gets \vec{y}]$, 
where $Y_1, \ldots Y_n$ is a (possibly) empty list of
distinct variables in $\V$ and $y_i \in \Rn(Y_i)$ for $i = 1, \ldots,
n$. 
This action represents an intervention that sets each $Y_i$
to the value $y_i$.  As we said earlier, we allow interventions only
on endogenous variables.  

The set $\A(\S)$ of \emph{actions} over $\S$ is defined recursively,
starting with 
the primitive actions, and closing under \textbf{if} \ldots
    \textbf{then} \ldots \textbf{else}, so that  if
$A,B \in \A$ and $\phi \in \L(\S)$, then
$$ \textbf{if } \phi  \textbf{ then } A \textbf{ else } B \in \A.$$
We take $\textbf{if } \phi  \textbf{ then } A $ to be 
an abbreviation of
$ \textbf{if } \phi  \textbf{ then } A \textbf{ else } \d[\,]$
($\d[\,]$ is the trivial action that sets no values).
Note that although we restrict interventions to interventions to endogenous
variables, the tests can involve exogenous variables.  For example,
even if interests rates are fixed (and hence taken to be exogenous),
we may want to say ``if the interest rate is 5\% then borrow \$1,000,000''
(where  we take the amount borrowed to be represented by an endogenous
variable).  

We assume that 
the decision maker has a preference relation (weak order) $\s_\S$
over the set of all actions in $\A(\S)$.  (We often omit the $\S$ if
it is clear from context or plays no role in the discussion.)
As usual, we write $A \sim B$ if $A \s B$ and $B \s A$.
Call a formula $\phi$ \emph{null} if
its conditional preference is trivial, that is, if $(\textbf{if } \phi
\textbf{ then } A)  \sim  (\textbf{if } \phi  \textbf{ then } B)$ for
all actions $A,B \in \A(\S)$. Call $\phi$ \emph{non-null} if it is not null. 

Each action $A \in \A(\S)$ defines a mapping $h_A: \at \to
\A^{\text{prim}}(\S)$ (where $\A^{\text{prim}}(\S)$ is the set of primitive actions).
 The functions $h_A$ are
defined recursively as follows:
$$h_{\d[\svc{Y}]}(\vec{a}) = \d[\svc{Y}]$$
and
$$ h_{\textbf{if } \phi  \textbf{ then } A \textbf{ else } B}(\vec{a}) =
\begin{cases}
    h_A(\vec{a}) \text{ if } \vec{a} \Rightarrow \phi \\
h_B(\vec{a}) \text{ if } \vec{a} \Rightarrow \neg \phi.
\end{cases}
$$

\subsection{Representation}
\label{sec:rep}
A \emph{subjective causal expected utility} representation understands
preferences as maximizing the expected utility of an action, relative
to uncertainty regarding the values of exogenous
variables. Specifically, the representation is governed by 
\begin{itemize}
    \item $M$: a model  that dictates the causal equations
    \item $\p$: a probability distribution on contexts
    \item $\mu$: a utility function on atoms.
\end{itemize}


An action (along with the model $M$) can be thought
 of as a operation that assigns values to all variables given a
 context.  
 Specifically given a model $M$, for each $A \in \A(\S)$, define
$\beta^M_A: \C(\S) \to \at$ as 
\begin{align*}
\beta^M_{A}(\u) = &\text{ the unique atom } \vec{a} \in \at  \\ 
& \text { such that }
 (M_{h_A(\vec{a}_{M,\u})},\vec{u})  \models \vec{a}.
\end{align*}

Unpacking this:  if the context is $\u$, then the initial assignment
of variables is expressed by the atom $\vec{a}_{M,\u}$, so
the primitive action prescribed by $A$ is $h_A(\vec{a}_{M,\u})$. Thus,
the final assignment of variables (i.e., after the intervention
$h_A(\vec{a}_{M,\u})$) is determined by the situation
$(M_{h_A(\vec{a}_{M,\u})},\vec{u})$: this is what is captured by
$\beta$.

\begin{definition}
  $(M,\p,\mu)$ \emph{is a subjective causal utility representation of
    $\s_\S$}, where $M=   (\S,\F)$ and $\S = (\U,\V,\R)$, if
$\p$ is a probability on $\Rn(\U)$, $\mu$ is a
  utility function on $\A(\S)$, and
  $A \s_{\S} B$ iff 
\begin{equation}
\label{eq:rep}
\sum_{\u \in \C(\S)} \mu(\beta^M_A(\u))\p(\u) \geq \sum_{\u \in
  \C(\S)} \mu(\beta^M_B(\u))\p(\u). 
\end{equation}
\end{definition}

In a representation, we think of contexts as states and atoms as
outcomes, so $\beta_A^M$ can be viewed as a function from states to
outcomes.  Thus, by associating $A$ with $\beta_A^M$, we can think of
$A$ as a function from states to outcomes---exactly how Savage views acts.

\section{Axioms}

In this section, we discuss the axioms that we need to ensure the
existence of a causal expected utility representation.

The first axiom is the cancellation axiom of
Blume, Easley, and Halpern \nciteyear{BEH06}.
To
define this we need the notion of a 
multiset, which can be thought of as a set that allows for multiple
instances of 
each of its elements; two multisets are equal just in case they contain the same elements with the same
multiplicities. For example, the multiset $\ms{a, a, a, b, b}$ is
different from the multiset $\ms{a, a, b, b, b}$: both 
multisets have five elements, but the multiplicities of $a$ and $b$
differ in the two multisets (assuming $a\neq b$). 

Using this notation, we can state the cancellation axiom.

\begin{ax}[Cancellation]
\label{ax:canc}
If $A_1 \ldots A_n, B_1 \ldots B_n \in \A$ and $\ms{h_{A_1}(\vec{a})
  \ldots h_{A_n}(\vec{a})} = \ms{h_{B_1}(\vec{a}) \ldots
  h_{B_n}(\vec{a})}$ for all $\vec{a} \in \at$, then $A_i \s B_i$ for
all $i < n$ implies $B_n \s A_n$. 
\end{ax}

As in the prior literature, cancellation allows us 
to construct a state space and an additively separable utility representation
of $\s_{\V}$. The state space that is constructed via the BEH
methodology does not have any causal structure.
The remaining axioms ensure that we can construct a subjective utility
representation that is a causal model.

As the remaining axioms speak directly to the structure of the
\emph{do} action, the following notation will be helpful:
For each atom $\vec{a} \in \at$, write $\d[\svc{Y}] \fixes_{\vec{a}} (Z
= z)$ as shorthand for the indifference relation 
\begin{equation}
  \label{eq:fixes}
  \textbf{if }  \vecc{a}  \textbf{ then } \d [\svc{Y}, Z
  \leftarrow z]  \sim \textbf{if } \vecc{a}  \textbf{
  then } \d [\svc{Y}].  
\end{equation}

This says, essentially, that in the context
characterized by the atom $\vec{a}$, setting $\vec{Y} = \vec{y}$
results in $Z=z$.  To understand why, 
notice that intervening to set $\vec{Y}$ to
$\vec{y}$ causes $Z$ to equal $z$ if and only if the further
intervention setting $Z$ to $z$ has no effect;
$\d[\svc{Y}] \fixes_{\vec{a}} (Z = z)$ can be thought of as capturing
this situation, as the decision maker sees no additional benefit to
the additional intervention on $Z$. Of course, it could be that $Z$ is
in fact set to some different value $z'$, but
the decision maker is indifferent between $Z=z$ and $Z=z'$,
given the values of the other variables.  This indifference is
irrelevant as far as the existence of a subjective utility
representation goes; we can 
we can simply treat the worlds where $Z=z$ and  $Z=z'$ identically.  

The next axiom, \emph{model uniqueness}, ensures that there is only
one model in our representation, rather than a distribution over models;
that is, the decision maker has no uncertainty about the structural
equations.  The axiom itself just says that, for each context
$\vec{u}$, there is at most one atom $\vec{a}$ compatible with $\vec{u}$ that is
non-null.  Intuitively, if $M$ represents $\s_\S$, then
$(M,\vec{u}) \models \vec{a}$.

\begin{ax}[Model Uniqueness]
\label{ax:ctx}
For each context $\u$, there exists a most one atom $\vec{a} \in \at$
such that $\vec{a} \Rightarrow (\U = \u)$ and $\vec{a}$ is
non-null. 
\end{ax}

The next axiom, Definiteness, ensures that there  must be some value $x
\in \Rn(X)$ such that $X=x$ after intervening to set $\vec{Y}$ to
$\vec{y}$.
This is essentially the content of the axiom of the same name
introduced by Galles and Pearl
\nciteyear{GallesPearl98}, also used by Halpern \nciteyear{Hal20}.

\begin{ax}[Definiteness]
\label{ax:def}
For each atom $\vec{a}$, vector of $\vec{Y}$ of endogenous variables,
$\vec{y} \in \Rn(\vec{Y})$,
and endogenous variable $X \notin \vec{Y}$, there
exists some $x \in \Rn(X)$ such that $\d[\svc{Y}] \fixes_{\vec{a}} (X = x)$. 
\end{ax}

The next axiom, \emph{Centeredness}, dictates that intervening to set
variables to their actual values does not change the values of other
variables; that is, trivial interventions are indeed
trivial.  It is named after \emph{centering} property considered by
Lewis \nciteyear{Lewis73}, which ensures that the the closest
world to a world $\w$ is $\w$ itself.
Let $\vec{a}|_{\vec{Y}}$ denote the restriction of the atom $\vec{a}$
to the conjuncts in $\vec{Y}$.

 \begin{ax}[Centeredness]
\label{ax:cent}
For each atom $\vec{a}$, vector of endogenous variables $\vec{Y}
$, and endogenous variable $X \notin \vec{Y}$,  we have
$\d[\vec{Y} \gets \vec{a}|_{\vec{Y}}] \fixes_{\vec{a}} (X = \vec{a}|_{X})$.   
\end{ax}

Finally, as we are interested in recursive causal models, we require
an axiom that forces the dependency order on endogenous variables to be acyclic.
Towards this, for $X,Y \in \V$, say that $X$ \emph{is unaffected by }
$Y$ (given $\vec{a}$) if 
    \begin{align}
    \begin{split}
\d[\svc{Z} ] \fixes_{\vec{a}} (X = x) &\text{ iff } \d[\svc{Z}, Y \leftarrow y] \fixes_{\vec{a}} (X = x) 
    \end{split}
    \label{eq:leadstocontra}
    \end{align}
  for all $\vec{Z} \in \V \setminus \{X,Y\}$,
    $\vec{z} \in \Rn(\vec{Z})$, $y \in \Rn(y)$, and $x \in \Rn(x)$. So
    $X$ is unaffected by $Y$ if there is no intervention on $Y$ that
    changes the decision maker's perception of $X$ (conditional on
        atom $\vec{a}$). If this relation does not hold, then $X$ is
        \emph{affected} by $Y$, written $Y \leadsto_{\vec{a}} X$.
Let $\leadsto \ = \cup_{\vec{a}} \leadsto_{\vec{a}}$. 
Our final axiom states that $\leadsto$ is acyclic; 
it is inspired by the corresponding axiom in \cite{Hal20}.

\begin{ax}[Recursivity]
\label{ax:rec}
$\leadsto$ is acyclic.
\end{ax}
Note that Axiom~A\ref{ax:rec} implies that $\leadsto_{\vec{a}}$ is
acyclic for each atom $\vec{a}$.

As we now show, a preference order satisfies these axioms iff it has
subjective causal expected utility representation.

\begin{theorem}
\label{thm:representation}
The preference order $\s_\S$ satisfies Axioms A\ref{ax:canc}--A\ref{ax:rec} if
and only if it has a subjective causal expected utility
representation $(M,\p,\mu)$. Moreover, if $\s_\S$ satisfies 
Axiom A\ref{ax:defp}\,$^*$, then $M$ is unique over the set of non-null contexts.
\end{theorem}

There may, in general, be many different subjective causal expected
utility representations of the same preference relations. This is
because, when the decision maker is indifferent between distinct
atoms, the preference relation cannot distinguish between them, and hence
cannot distinguish between structural equations
that yield distinct but equally valued atoms. 

\shortv{
By strengthening Axiom \ref{ax:def} to require that the indifference
condition holds for a unique assignment of variables, it is possible
to avoid this problem of multiple representations. This stronger
axiomatization ensures there is a unique model that rationalizes the
decision maker's preference. We refer the interested reader to the
appendix for the formal axiom and uniqueness claim. 
}

\fullv{
Consider the following strengthening of A\ref{ax:def}:
\addtocounter{ax}{-3}
\begin{axprime}[Strong Definiteness]
\label{ax:defp}
For each non-null atom $\vec{a}$, vector of $\vec{Y}$ of endogenous variables,
$\vec{y} \in \Rn(\vec{Y})$,
and endogenous variable $X \notin \vec{Y}$, there
exists a \emph{unique} $x \in \Rn(X)$ such that $\d[\svc{Y}] \fixes_{\vec{a}} (X = x)$. 
\end{axprime}

Strengthening A\ref{ax:def} to A\ref{ax:defp}$^*$, is both necessary
and sufficient to avoid multiple representations.

\begin{definition}
  A subjective causal utility representation of a preference relation of
$(M',\p',\mu')$, non-null contexts $\u$,  and formulas $\phi \in \L^+(\S)$,
$$ (M,\u) \models \phi \qquad \text{ if and only if} \qquad (M',\u)
  \models \phi.$$ 
\end{definition}

\begin{theorem}
\label{thm:uniq}
A subjective causal utility representation $(M,\p,\mu)$ of
$\s_\S$ is \emph{identified} if and only if $\s_\S$ satisfies 
Axiom A\ref{ax:defp}$^*$.
\end{theorem}

}

\fullv{
\section{Proof of Theorems \ref{thm:representation} and \ref{thm:uniq}}
}
\shortv{
\section{Proof of Theorem \ref{thm:representation}}
}
\label{sec:repPrf}
\fullv{
Before getting into the details of the proof, we describe it at a
higher level.  
}
\shortv{
    This section contains an overview of the proof of our
    representation theorem. Some details of the proof, 
namely, proofs of intermediate lemmas, 
    are omitted and can be found in the appendix.
    We begin by describing the proof at a higher level.
}

Given a preference order $\s_{\S}$, the first step is to
construct a Lewis-style model $M_{\s_{\S}}$ that represents
  $\s_\S$.  We want $M_{\s_{\S}}$ to be recursive in a
precise sense.
So for each atom $\vec{a} \in \at$, we define a strict
linear order $<_{\vec{a}}$ of $\at$, intended to represent the
closeness  of atoms to $\vec{a}$.
We then show that, for each atom
$\vec{a}$, intervening to set $\svc{Y}$ has the same consequence as
intervening to set \emph{all} variables to their values in the
$<_{\vec{a}}$-closest atom to $\vec{a}$ in which $\vec{Y}=\vec{y}$
holds (see Lemma \ref{lem:sameIntervention}).  This is the only property that
BH used to prove their
representation theorem, which used Lewis-style models (see their Lemma
5 and Theorem 1).  It follows that we get the desired Lewis-style
representation of  
$M_{\s_{\S}}$.  Moreover, in  $M_{\s_{\S}}$, $<_{\vec{a}}$
does in fact represent the closeness of atoms to $\vec{a}$. These
observations allow us to appeal to results of Halpern
\nciteyear{Hal40}, and convert $M_{\s{\S}}$ to a causal model that
also represents $\s_\S$.   

So, to begin, we want to define $<_{\vec{a}}$. 
Axiom A\ref{ax:rec} ensures that we can extend
$\leadsto$ (and hence $\leadsto_{\vec{a}}$ for each atom $\vec{a}$) 
to a strict linear order on $\U \cup \V$, denoted
$\dot{\leadsto}$, such that for all $U \in
\U$ and $X \in \V$ we have $U \dot{\leadsto} X$. For atoms
$\vec{b}, \vec{c} \in
\at$, let $Y_{\vec{b}, \vec{c}} \in \U\cup
\V$ be the $\dot{\leadsto}$-minimal variable on which
$\vec{b}$ and $\vec{c}$ disagree. For each variable $X \in \U \cup \V$, define
$\ll^X_{\vec{a}}$ to be some fixed strict linear order of $\Rn(X)$
whose minimal element is the value of $X$ in atom $\vec{a}$.
We define $<_{\vec{a}}$ to be, loosely speaking, a
lexicographic order over atoms, ordering first over variables, using
$\dot{\leadsto}$, and then over values, using
$\ll_{\vec{a}}$. We discuss the intuition behind this order after
giving its
definition: 

Take $<_{\vec{a}}$ to be a strict linear order over $\at$ such that
%
\begin{enumerate}[label=\textsc{or\arabic*.}]
\item  \label{o.1}  $\vec{b} <_{\vec{a}} \vec{c}$ if
 $Y_{\vec{a}, \vec{c}} \ \dot{\leadsto} Y_{\vec{a}, 
      \vec{b}}$, and 
 \item  \label{o.2} If $Y_{ \vec{a}, \vec{b}} = Y_{
      \vec{a}, \vec{c}}$, 
                then let $Y$ denote $Y_{ \vec{b}, \vec{c}}$,
        let $\vec{Z}$ be the set of endogenous
                variables (strictly) $\dot{\leadsto}$-less than $Y$,
      and let $\vec{z} = \vec{b}|_{\vec{Z}} (= \vec{c}|_{\vec{Z}})$.
      Then $\vec{b} <_{\vec{a}} 
      \vec{c}$ if 
     $\d[\svc{Z}\,] \fixes_{\vec{a}} \vec{b}|_{Y}$ and either
        \begin{enumerate}[label=(\roman*)]
             \item \label{o.2b} not $\d[\svc{Z}\,] \fixes_{\vec{a}} \vec{c}|_{Y}$, or
            \item \label{o.2a} $\d[\svc{Z}\,] \fixes_{\vec{a}}
              \vec{c}|_{Y}$ and
              $y_{\vec{b}} \ll^{Y} _{\vec{a}} y_{\vec{c}}$, where
              $y_{\vec{b}}$ (resp., $y_{\vec{c}}$ denotes the value of
              $Y$ in $\vec{b}$ (resp., $\vec{c}$).
         \end{enumerate} 
\end{enumerate}
In general, \eqref{o.1} and \eqref{o.2} do not completely determine the
order; there may be 
atoms for which neither \eqref{o.1} nor \eqref{o.2} hold. The
remainder of the order can be completed arbitrarily. 

%
A few points of intuition regarding this order. In the Lewis-style
model $M_{\s_\S}$ that we construct, the states are represented
by pairs of atoms.   The operator $\fixes_{\vec{a}}$ lets us probe the 
structural equations through the effect of interventions.
When considering which of $\vec{b}$ or $\vec{c}$ is closer to $\vec{a}$
according to $\fixes_{\vec{a}}$, 
there are two criteria of lexicographic importance: (i)
which of $\vec{b}$ or $\vec{c}$ coincides with $\vec{a}$ longer
(where \emph{longer} means ``for more variables, starting with the exogenous
context and proceeding via $\dot{\leadsto}$''), and  (ii) if both
$\vec{b}$ and $\vec{c}$ deviate from $\vec{a}$ at the same
variable $Y$, does one coincide with the counterfactual assessment of
$\vec{a}$ given by $\fixes_{\vec{a}}$, and if both do, is one closer
to $\vec{a}$ than the other according to $\ll^{Y} _{\vec{a}}$.
These two criteria are
captured by \eqref{o.1} and \eqref{o.2}, respectively. 

An example may help clarify: Let $\vec{a}$ be the atom $U=0 \land X=0
\land Y=0 
\land Z=0$, and assume
that the order  $\dot{\leadsto}$ on variables is $U \dot{\leadsto} X
\dot{\leadsto} Y \dot{\leadsto} Z$. Now consider the following
atoms: 
\begin{align*}
  \vec{b} &= U= 0 \land X=0 \land Y=0 \land Z=1 \\
    \vec{c} &= U= 0 \land X= 0 \land Y=1 \land Z=1 \\
    \vec{c}^{\, \prime}& = U= 0 \land X=0 \land Y=1 \land Z=0.
\end{align*}
%
\eqref{o.1} states that if an atom coincides with $\vec{a}$ longer,
it is closer to $\vec{a}$, so $\vec{b}$ is closer to $\vec{a}$ than either
$\vec{c}$ or $\vec{c}^{\, \prime}$.
Since $\vec{c}$ and $\vec{c}^{\, \prime}$ agree up to $Z$,
their closeness to $\vec{a}$ depends on the value(s) of $z$ such that
$\d[X \gets 0, Y \gets 1] \fixes_{\vec{a}} (Z=z)$.  In particular, if
we can have $z=1$ (which is the value of $Z$ according to $\vec{c}$)
and not $z=0$, then 
$\vec{c}$ is closer to $\vec{a}$ that $\vec{c}^{\, \prime}$.  Intuitively, this
is because, according to
$\fixes_{\vec{a}}$, if we set $X$ to 0 and $Y$ to 1, $Z$ would be 1,
so $\vec{c}$ is consistent with the equations we plan to use in the
context encoded by $\vec{a}$, while $\vec{c}^{\, \prime}$ is not.
Similarly, if we can have $z=0$ and not $z=1$, then $\vec{c}^{\, \prime}$
is closer.  Because of the possibility of indifference, both $z=0$ and
$z=1$ could be consistent with $\fixes_{\vec{a}}$.  In this case, the
order $\ll^{Y} _{\vec{a}}$ serves as a consistent method of
breaking ties.\footnote{$\ll^{Y} _{\vec{a}}$ is arbitrary except that
its initial element is the value of $Y$ in $\vec{a}$, reflecting the
fact that beyond 
coinciding with $\vec{a}$ if possible, this tie-breaking is
arbitrary. Nonetheless, an order still needs to be fixed to ensure
tie-breaking is consistent across different interventions.} 
 
 \bigskip

Let $\vec{a}_{\d[\svc{Y}]}$ denote  the $<_{\vec{a}}$-minimal atom satisfying
$\vec{Y}= \vec{y}$.
 
\begin{lemma}
\label{lem:minatom}
Suppose that $\vec{Y}$ is a set of endogenous variables, $X \in \V
\setminus \vec{Y}$,  
and $\vec{Z}_X \subseteq \V$ is the (possibly empty) set of
endogenous variables strictly $\dot{\leadsto}$-less than
$X$.  
Then we have
\begin{enumerate}[label=\textsc{M\arabic*.}]
    \item \label{m.minY} $\vec{a}_{\d[\svc{Y}]}|_{\U} = \vec{a}|_{\U}$, and
          \item \label{m.fixes} The value of $X$ in $\vec{a}_{\d[\svc{Y}]}$ is the
      $\ll^{X}_{\vec{a}}$-minimal element of  
      \begin{gather*}
      m(\vec{a}, X) = \\ \big\{ x \in \Rn(X) : 
            \d[\vec{Z}_X \gets \vec{a}_{\d[\svc{Y}]}|_{\vec{Z}_X}]
            \fixes_{\vec{a}} (X= x)\big\}.
      \end{gather*}

\end{enumerate}
\end{lemma}

\fullv{
 \begin{proof}
   Since $\U$ forms the initial segment of $\dot{\leadsto}$, it
is follows from \eqref{o.1} that \eqref{m.minY} must hold.
In more detail, if $\vec{a}_{\d[\svc{Y}]}|_{\U} \ne \vec{a}|_{\U}$, then
let $\vec{b}$ be such that $\vec{a}_{\d[\svc{Y}]}|_{\V} = \vec{b}|_{\V}$
and $\vec{b}_{\U} = \vec{a}|_{\U}$.  Since $\U$ forms the initial
segment of $\dot{\leadsto}$, it follows from \eqref{o.1}
that $\vec{b} <_{\vec{a}} \vec{a}_{\d[\svc{Y}]}|_{\U}$, yet
$\vec{b}$ satisfies  $\vec{Y} = \vec{y}$. This is a contradiction.

 To see that \eqref{m.fixes} holds, suppose by way of contradiction
 that it does not hold.  Then there exists a variable $X$ such that
the value of $X$ in  $\vec{a}_{\d[\svc{Y}]}$ is not the
$\ll^{X}_{\vec{a}}$-minimal element of
$m(\vec{a}, X)$ (note that $m(\vec{a}, X)$ is non-empty by Axiom
 A\ref{ax:def}).   
Let $\vec{b}$ be the atom that coincides with
$\vec{a}_{\d[\svc{Y}]}$ for all variables except $X$, and the value of $X$
in $\vec{b}$ is the
$\ll^{X}_{\vec{a}}$-minimal element of $m(\vec{a}, X)$.
There are two cases: 
 \begin{itemize}
 \item
   $\vec{a}|_{\vec{Z}_X} = \vec{a}_{\d[\svc{Y}]}|_{\vec{Z}_X}$:
   From Axiom A\ref{ax:cent}, it follows that $\d[\vec{Z}_X \leftarrow
    \vec{z}_X] \fixes_{\vec{a}} \vec{a}|_{X} $. By the definition of
       $\ll^{X}_{\vec{a}}$, the value of $X$ in $\vec{a}$ is the
       $\ll^{X}_{\vec{a}}$-minimal element of $m(\vec{a}, X)$. By
   assumption, the value of $X$ in  $\vec{a}_{\d[\svc{Y}]}|_{X}$ is not the
   $\ll^{X}_{\vec{a}}$-minimal element of $m(\vec{a}, X)$, while the
   value of $X$ in $\vec{b}$ is the  $\ll^{X}_{\vec{a}}$-minimal
   element of $m(\vec{a}, X)$.  It follows that $\vec{b}  <_{\vec{a}}
   \vec{a}_{\d[\svc{Y}]}$, a contradiction.
     \item $\vec{a}|_{\vec{Z}_X} \neq
       \vec{a}_{\d[\svc{Y}]}|_{\vec{Z}_X}$: Since $\vec{a}|_{\vec{Z}_X} \neq
       \vec{a}_{\d[\svc{Y}]}|_{\vec{Z}_X} = \vec{b}|_{\vec{Z}_X}$, we have
%
       $Y_{\vec{a}, \vec{a}, \vec{a}_{\d[\svc{Y}]}} =
       Y_{\vec{a}, \vec{a}, \vec{b}}$, so by
         \ref{o.2}, $\vec{b} <_{\vec{a}} \vec{a}_{\d[\svc{Y}]}$, and again
         we have a contradiction. 
\end{itemize}
  \end{proof}

\begin{lemma}
  \label{lem:extendVars}
  For all atoms $\vec{a} \in \at$, (disjoint) vectors of endogenous variables
$\vec{Y}$ and  $\vec{Z}$, $\vec{y} \in \Rn(\vec{Y})$,
  $\vec{z} \in \Rn(\vec{Z})$, endogenous variables $X  
\notin \vec{Z}\cup\vec{Y}$ such that for all $Y\in \vec{Y}$, $X
\dot{\leadsto} Y$, and $x \in \Rn(X)$, we have
\begin{align*}
&\d[\svc{Z}\,] \fixes_{\vec{a}} (X= x) \text{ iff } \\
&\d[\svc{Z},\svc{Y}] \fixes_{\vec{a}} (X= x). 
\end{align*}

\end{lemma}

\begin{proof}
  By assumption, for all $Y \in \vec{Y}$, we have $Y
 {\centernot{\leadsto}}_{\vec{a}}\, 
  X$. The lemma then follows from  a simple induction argument on the
variables of $\vec{Y}$, appealing to \eqref{eq:leadstocontra}  in each
step. 
\end{proof}
}

The following lemma shows that, starting with $\vec{a}$, the effect of
setting $\vec{Y}$ to $\vec{y}$ is the same as setting all the
endogenous variables to the value that results from setting $\vec{Y}$
to $\vec{y}$.  Intuitively, this is because for a variable $X \notin
\vec{Y}$, we are setting it to the value it already has after setting
$\vec{Y}$ to $\vec{y}$, so nothing further changes.

\begin{lemma}
\label{lem:sameIntervention}
For all atoms $\vec{a} \in \at$, vectors of endogenous variables $\vec{Y}$, and
$\vec{y} \in \Rn(\vec{Y})$, we have that 
$$
\textbf{if } \vecc{a}  \textbf{ then } \d [\svc{Y}]
\sim \textbf{if } \vecc{a}  \textbf{ then } \d [\V
  \leftarrow \vec{a}_{\d[\svc{Y}]}|_{\V}].  
$$
\end{lemma}

\fullv{
\begin{proof}
  We prove this by induction on the variables in $\V$ with respect to  the order
   $\dot{\leadsto}$. Fix some $X \in \V\setminus \vec{Y}$, 
let $\vec{Z} \subseteq \V$ denote the (possibly empty) vector of
endogenous variables strictly $\dot{\leadsto}$-less than $X$, 
and let $\vec{z}$ denote the values of the variables $\vec{Z}$ in
$\vec{a}_{\d[\svc{Y}]}$.
%
To simplify notation, let $x$ denote the value of $X$ in
$\vec{a}_{\d[\svc{Y}]}$, let $\vec{Y}' \subseteq \vec{Y}$ consist of those variables
$\dot{\leadsto}$-greater than $X$, and let $\vec{y}^{\,
  \prime}$ be the restriction of $\vec{y}$ to $\vec{Y}'$.
We now show that if 
\begin{equation}
\label{eq:induct_hypo2}
\textbf{if } \vec{a}  \textbf{ then } \d [\svc{Y}] \\ \sim \\ \textbf{if } \vec{a}  \textbf{ then } \d [\svc{Z}, \vec{Y}' \leftarrow \vec{y}^{\, \prime}] 
\end{equation}
then 
\begin{gather}
\label{eq:induct_conseq}
\textbf{if } \vecc{a}  \textbf{ then } \d [\svc{Y}] \nonumber \\
\sim \\ 
\nonumber
\textbf{if } 
\vecc{a}  \textbf{ then } \d [\svc{Z}, \vec{Y}' \leftarrow \vec{y}^{\,
        \prime}, X \leftarrow x].  
\end{gather}
From \eqref{m.fixes}, it follows that $\d[\svc{Z}\,] \fixes_{\vec{a}} (X=
x)$. Applying Lemma \ref{lem:extendVars}, we obtain 
\begin{equation}
\label{eq:x_fixed_yprime}
\d[\svc{Z}, \vec{Y}' \leftarrow \vec{y}^{\, \prime}] \fixes_{\vec{a}} (X= x).
\end{equation}
Thus,
\begin{align*}
&\textbf{if } \vecc{a}  \textbf{ then } \d [\svc{Z}, \vec{Y}' \leftarrow \vec{y}^{\, \prime},X \leftarrow x] \\
\sim \quad &\textbf{if } \vecc{a}  \textbf{ then } \d [\svc{Z},
  \vec{Y}' \leftarrow \vec{y}^{\, \prime}] 
\\ 
\sim \quad 
&\textbf{if } \vecc{a}  \textbf{ then } \d [\vec{Y}
    \leftarrow \vec{y}], 
\end{align*}
where the first indifference comes from from \eqref{eq:x_fixed_yprime}
and the definition of $\fixes_{\vec{a}}$), 
and the second from \eqref{eq:induct_hypo2}.
  \end{proof}
}

BH show in their Theorem 1 
that (in our terminology), given a 
preference $\s_\S$  satisfying their Lemma 5 (which is a direct
translation of our Lemma \ref{lem:sameIntervention}) and the Cancellation axiom, and a family
$<_a$ of linear orders, one for each atom, that they can construct a
Lewis-style model that represents 
$\s_\S$.  We outline the construction below, as well as defining
formally what  it means
for a Lewis-style model to represent $\s_\S$.

The states in the Lewis-style model $M$ are pairs $(\vec{a},\vec{a}')$ of
atoms.  Roughly speaking, since we are trying to capture the effect of
interventions, we can think of $\vec{a}$ as the current world (the
result of performing the intervention) and
$\vec{a}'$ as the world before the intervention was performed.  The
truth of formulas in $\L(\S)$ is completely determined by the first
component. That is, $(M,(\vec{a},\vec{a}'))
\models Y=y$ if $Y=y$ is a conjunct of $\vec{a}$; we extend to
negation, conjunction, and negation in the standard way.  To extend
this semantics to $\L^+(\S)$, we need a closeness relation for each
pair $(\vec{a},\vec{a}')$.  Let $\sqsubset_{\vec{a},\vec{a}'}$ 
be an arbitrary family of strict linear orders on  $\at \times \at$ such that
$$
\pair{b} \sqsubset_{\pair{a}} \pair{c} \text{ whenever } \vec{b}
<_{\vec{a}} \vec{c} 
$$
and
$$
\langle \vec{b}, \vec{a} \rangle \sqsubset_{\pair{a}} \pair{b} \text { for all } \vec{b}^{\, \prime} \neq \vec{a}
$$
(so that $\langle \vec{a},\vec{a} \rangle$ is the minimal element of
$\sqsubset_{\pair{a}})$.   
This completes the description of $M$.

Let $\textsc{min}_{\pair{a}}(\d[\vec{Y} \gets \vec{y}]) \in \at \times \at$ denote the
$\sqsubset_{\pair{a}}$-minimal state such that 
$(M,\textsc{min}_{\pair{a}}(\d[\vec{Y}\gets \vec{y}])) \models \vec{Y} \gets \vec{Y}$.
Thus, $(M,\pair{a}) \models \d[\svc{Y}]\phi$ iff 
 $(M, \textsc{min}_{\pair{a}}(\d[\svc{Y}])) \models \phi$.   It is easy
to check that 
\begin{equation}
\label{eq:mintoclosest}
\textsc{min}_{\pair{a}}(\d[\vec{Y} \gets \vec{y}]) =
\langle \vec{a}_{\vec{Y} \gets \vec{y}},\vec{a}\rangle.
\end{equation}
Thus, the first
component of the minimal state that results from performing the intervention $\vec{Y}
\gets \vec{y}$ in a state of the form $\langle
\vec{a},\vec{a}'\rangle$ encodes the atom that results when the
intervention is performed starting at $\vec{a}$, while the second
component keeps track of the atom we started with.

BH define an analogue of our function $h_A$, which we denote $h_A'$,
that associates with each of their actions an action of the form
$\d[\svc{Y}]$.\footnote{BH allow for (but do not require) a
richer set of interventions, they allow $\d[\phi]$ for any
consistent formula $\phi$ in the language. }  For a state $\omega$, the
analogue of $\beta_A^M(\omega)$ 
is $\textsc{min}_{\omega}(h_A'(\omega))$.%
\footnote{The states for us are contexts, 
whereas for BH they are pairs of atoms, so using the same symbol
$\omega$ for states is somewhat of an abuse of notation; we hope that
our intention is clear here.}
BH show that there exists a probability measure $\p$ and utility
function $\mu$, both defined on the states of $M$, such that
  $A \s_{\S} B$ iff 
\begin{equation}
  \label{eq:rep}
  \begin{array}{lll}
&\sum\limits_{\substack{\pair{a} \in \\ \at \times \at}} \p(\pair{a}) \cdot 
\mu(\textsc{min}_{\pair{a}}(h'_A(\pair{a}))) \\ \ge
&\sum\limits_{\substack{\pair{a} \in \\ \at \times \at}} \p(\pair{a}) \cdot 
\mu(\textsc{min}_{\pair{a}}(h'_B(\pair{a}))).
\end{array}
  \end{equation}
This is the sense in which the Lewis-style model $M$ represents $\s_\S$.

Since BH's Lemma 5 is a direct translation of our Lemma \ref{lem:sameIntervention}, it follows
from our axioms as well that this Lewis-style model $M$ represents $\s_\S$.
But we are not done; we need to get a causal model that represents $\s_\S$.  
So we want to convert the Lewis-style model constructed by BH to a
causal model $M^C = (\S^C,\F^C)$, 
and construct an appropriate probability $\p^C$ on the contexts of
$M$, and utility $\mu^C$ on $\A(\S)$.

\commentout{
Before getting to this, for any pair of atoms, consider the semantics
for $\L(\S)$ that simply ignores the second component, that is:  
\begin{equation}
\label{eq:lewisSem}
\pair{a} \models (X=x) \text{ iff }\vecc{a}|_{X} = (X =x)
\end{equation}
and $\models$ extends to $\L(\S)$ in the standard way.

We can turn this into a counterfactual ``Lewis-Style" model over the
expanded language $\L^+$ that contains $\L(\S)$ and all formula of the
form $\d[\svc{Y}]\phi$, where $\phi \in \L(\S)$. The interpretation
here is that $\phi$ is true in the counterfactual world where
$\vec{Y}$ is intervened upon and set to the values $\vec{y}$. To give
semantics to these new formulae, we define a parameterized family of
closeness operators $\sqsubset_{\vec{a}}$.  Let $\sqsubset_{\vec{a}}$
be any linear ordering over  $\at \times \at$ that satisfies
Let $\sqsubset_{\vec{a}}$
be a linear order on $\at \times \at$ that satisfies
$$
\pair{b} \sqsubset_{\vec{a}} \pair{c} \text{ whenever } \vec{b} <_{\vec{a}} \vec{c}
$$
and
$$
\langle \vec{b}, \vec{a} \rangle \sqsubset_{\vec{a}} \pair{b} \text { for all } \vec{b}^{\, \prime} \neq \vec{a}.
$$
and define $\textsc{min}_{\vec{a}}(\d[\svc{Y}]) \in \at \times \at$ as the $\sqsubset_{\vec{a}}$-minimal element, such that $\textsc{min}_{\vec{a}}(\d[\svc{Y}]) \models [\vec{Y} = \vec{y}]$.
Using this, we can extend our semantics to modal formulae:
\begin{equation}
\label{eq:modalext}
\pair{a} \models \d[\svc{Y}]\phi \text{ iff } \textsc{min}_{\vec{a}}(\d[\svc{Y}]) \models \phi
\end{equation}

Appealing to Theorem 1 of \cite{bjorndahl2023sequential} of BH, we obtain the pair $(p,u)$, a fully supported probability distribution and utility function, respectively, over pairs of atoms, such that
\begin{equation}
\label{eq:BHrep}
U: A \mapsto  \quad \sum_{\vec{a} \in \at} \sum_{\vec{a}^{\, \prime} \in \at} p(\pair{a}) \cdot u(\textsc{min}_{\vec{a}}(h_A(\vec{a})))
\end{equation}
is utility representation over $\s$. 
Moreover, BH show that
\begin{equation}
\label{eq:reducestate}
\textsc{min}_{\vec{a}}(\d[\svc{Y}]) = \langle \vec{a}_{\d[\svc{Y}]}, \vec{a} \rangle
\end{equation}
where $\vec{a}_{\d[\svc{Y}]}$ is the $<_{\vec{a}}$-minimal atom such that $[\vec{Y} = \vec{y}]$ holds (as introduced in Lemma \ref{lem:minatom}).
}

The first step is easy.  We start with a signature $\S$ that
determines the language; we take $\S^C = \S$.  
By Axiom A\ref{ax:ctx}, it follows that
for each context $\u \in \C(\S)$, there at most one atom such that
such that $\vec{a} \Rightarrow (\U = \u)$ and $\vec{a}$ is non-null. 
Let $\C^\dag$ denote the set of contexts for
which such a non-null atom exists, and for each context $\u \in
\C^\dag$, let $\vec{a}_{\u}$ denote this non-null atom.
We define $\p^C$ using the probability $\p$ in the model $M$ provided
by BH so that it has support $\C^\dag$: 
$$
\p^C(\u) = \frac{\sum_{\vec{a}^{\, \prime} \in \at} \p(\langle \vec{a}_{\u}, \vec{a}^{\, \prime} \rangle)}{\sum_{\vec{u}^{\, \prime} \in \C^\dag}\sum_{\vec{a}^{\, \prime} \in \at} \p(\langle \vec{a}_{\vec{u}^{\, \prime}}, \vec{a}^{\, \prime} \rangle)}.
$$
%


Observe that, since the truth of
formulas in $M$ at state $\pair{a}$ is fully determined by $\vec{a}$, $h'_A(\pair{a})$ does not depend on $\vecc{a}'$. In particular, $h'_A(\pair{a}) = h_A(\vec{a})$. Similarly, by \eqref{eq:mintoclosest},
$\textsc{min}_{\pair{a}}(\cdot)$ does not depend on $\vecc{a}'$. We can therefore rewrite the
BH representation \eqref{eq:rep} as
\begin{equation}
  \label{eq:nurep}
  \begin{array}{lll}
&\sum_{\u \in \C^\dag} \p^C(\u) \cdot 
\mu(\textsc{min}_{\langle \vec{a}_{\u}, \vec{a}_{\u} \rangle}(h_A( \vec{a}_{\u} ))) \\ \ge
&\sum_{\u \in \C^\dag} \p^C(\u) \cdot 
\mu(\textsc{min}_{\langle \vec{a}_{\u}, \vec{a}_{\u} \rangle}(h_B(\vec{a}_{\u} ))).
\end{array}
  \end{equation}

We next want to show that there exists a set of equations $\F^C$ such that $M^C = (\S^C,\F^C)$  satisfies the same formulas as $M$. Specifically:

\begin{lemma}
\label{lem:biequiv}
There exists a set of equations $\F^C$ 
such that for all $\phi \in \L(\S)$ and all interventions $\vec{Y}
\gets \vec{y}$, we have that
\begin{align}
\label{eq:equivmodel1}
\nonumber
&(M^C_{\d[\svc{Y}]},\u) \models \phi  \quad \mbox{ iff } \\
&(M,\textsc{min}_{\langle \vec{a}_{\u}, \vec{a}_{\u}
  \rangle}(\d[\svc{Y}])) \models \phi.  
\end{align}
Note that taking $\vec{Y} = \emptyset$ gives a special case of
(\ref{eq:equivmodel1}): 
$(M^C,\u) \models \phi$ iff 
$(M,\langle  \vec{a}_{\u}, \vec{a}_{\u} \rangle) \models \phi$.
\end{lemma}
  
\fullv{
\begin{proof}
Let $Y_1, \ldots, Y_k$ be the ordering on endogenous variables
according to $\dot{\leadsto}$. 
We define $\F^C$ for the endogenous variables by induction $j$ so that
(\ref{eq:equivmodel1}) for all formulas $Y_j = y_j$.

For the base case, note that the value of the variable $Y_1$ does not
depend on the 
values of any other endogenous variable.  If is straightforward to
define $F_{Y_1}$ so that $F_{Y_1}(\vec{u}) = y_1$ iff
$a_{\vec{u}} \Rightarrow Y_1=y_1$.  It easily follows that 
$(M^C,\u) \models Y_1 = y_1$ iff 
$(M,\langle  \vec{a}_{\u}, \vec{a}_{\u} \rangle) \models Y_1 = y_1$.
Moreover, if $Y_1 \notin \vec{Y}$, then 
$(M^C_{\d[\svc{Y}]},\u) \models Y_1 = y_1$ iff
$(M^C,\u) \models Y_1 = y_1$.  
By Lemma~\ref{lem:extendVars} (taking $\vec{Z} = \emptyset$ and $X = Y_1$)
it follows that $(M,\textsc{min}_{\langle \vec{a}_{\u}, \vec{a}_{\u}
  \rangle}(\d[\svc{Y}])) \models Y_1 = y_1$ iff $(M,\langle \vec{a}_{\u},
\vec{a}_{\u}) \models Y_1 = y_1$.  On the other hand, if $Y_1 \in
\vec{Y}$, then $(M^C_{\d[\svc{Y}]},\u) \models Y_1 = y_1$ iff
$\vec{y}\mid_{Y_1} = y_1$ and similarly, $(M,\textsc{min}_{\langle \vec{a}_{\u}, \vec{a}_{\u}
  \rangle}(\d[\svc{Y}])) \models Y_1 = y_1$ iff  $\vec{y}\mid_{Y_1} = y_1$.
Thus, we get the desired result in the base case.

For the inductive case, suppose that we have proved the result for
$Y_1, \ldots, Y_j$.  We can define $F_{Y_{j+1}}$ to be function of the
  exogenous variables and $Y_1, \ldots, Y_j$ such that for all $k \in
  \{1,\ldots, j\}$, and $y_i \in \Rn(Y_i)$ for $i = 1,\ldots, j$, we
  have that $F_{Y_{j+1}}(\vec{u}, y_1, \ldots, y_j) = y_{k+1}$ iff
$(M,\textsc{min}_{\langle \vec{a}_{\u}, \vec{a}_{\u}
  \rangle}(\d[Y_1 \gets y_1, \ldots, Y_k \gets y_k]) \models Y_{k+1}
= y_{k+1}$.  Now for arbitrary $\vec{Y}$, if $Y_{j+1} \notin \vec{Y}$,
let $\vec{Y}' = \vec{Y} \cap
\{Y_1, \ldots, Y_j\}$.  Let $\vec{y}' = \vec{y}|_{\vec{Y}'}$.  Writing
$\vec{Y}''$ for $\{Y_1, \ldots, Y_j\}$, let $\vec{y}''$ be such that
$(M^C_{\d[\vec{Y}' \gets \vec{y}']},\u) \models \vec{Y}'' =
\vec{y}''$.  Then it follows from the induction hypothesis that 
 $(M,\textsc{min}_{\langle \vec{a}_{\u}, \vec{a}_{\u}
  \rangle}(\d[\vec{Y}' \gets \vec{y}'])) \models \vec{Y}'' =
\vec{y}''$.
Moreover, it easily follows that $(M^C_{\d[\svc{Y}]},\u) \models
Y_{j+1} = y_{j+1}$ iff $(M^C_{\d[\vec{Y}'' \gets\vec{y}'']},\u) \models
Y_{j+1} = y_{j+1}$ and 
 $(M,\textsc{min}_{\langle \vec{a}_{\u}, \vec{a}_{\u}
  \rangle}(\d[\vec{Y} \gets \vec{y}])) \models Y_{j+1} =
y_{j+1}$ iff 
 $(M,\textsc{min}_{\langle \vec{a}_{\u}, \vec{a}_{\u}
  \rangle}(\d[\vec{Y}'' \gets \vec{y}''])) \models Y_{j+1} =
y_{j+1}$.  Since, by the definition of $F_{Y_{j+1}}$,
$(M^C_{\d[\vec{Y}'' \gets\vec{y}'']},\u) \models Y_{j+1} = y_{j+1}$ 
iff
 $(M,\textsc{min}_{\langle \vec{a}_{\u}, \vec{a}_{\u}
  \rangle}(\d[\vec{Y}'' \gets \vec{y}''])) \models Y_{j+1} =
y_{j+1}$, the result easily follows.  The argument if $Y_{j+1} \in
\vec{Y}$ is the same as in the base 
case.  

The result for arbitrary
formulas $\phi \in \L(\S)$ is immediate (since conjunction,
disjunction, and negation work the same way in both $M$ and $M^C$).
\end{proof}
}

From \eqref{eq:mintoclosest}, it follows that $\textsc{min}_{\langle
  \vec{a}_{\u}, \vec{a}_{\u} 
  \rangle}(\d[\svc{Y}]) = \langle
(\vec{a}_{\u})_{\d[\svc{Y}]},\vec{a}\rangle$; moreover,
$(M, \langle ( \vec{a}_{\u})_{\d[\svc{Y}]},\vec{a}\rangle) \models
(\vec{a}_{\u})_{\d[\svc{y}]}$. Thus, by Lemma \ref{lem:biequiv},  
$(M^C_{\d[\svc{y}]},\u)  \models
(\vec{a}_{\u})_{\d[\svc{y}]}$. Moreover, since $(M^c, \u) \models
\vec{a}_{\u}$, that is, $\vec{a}_{M^C, \u} = \vec{a}_{\u}$, we have 
  \begin{equation}
  \label{eq:betaviah}
  \beta^{M^C}_A(\u) = (\vec{a}_{\u})_{h_A(\vec{a}_{\u})}.
  \end{equation}

\commentout{
For each pair $\pair{a}$ in the support of $p$, we have that the order
$\dot{\leadsto}_{\vec{a}}$ is \emph{recursive} (with respect to the
order $\sqsubset_{\vec{a}}$over atoms). Specifically, that is,  
if $X \dot{\leadsto}_{\vec{a}} W$, then for all $\vec{Y} \subseteq (\U \cup \V)\setminus\{X,Z\}$ and values for $\vec{y}$ for the variables in $\vec{Y}$ and $w \in \Rn(W)$, 
$$
\pair{a} \models \d[\svc{Y}](X = x)\quad \text{ iff }\quad \pair{a} \models \d[\svc{Y},W \leftarrow w](X = x).
$$

  \begin{proof}
We have:
\begin{align*}
&\pair{a} \models \d[\svc{Y}](X = x) \\
\text{ iff } \quad &\langle \vec{a}_{\d[\svc{Y}]}, \vec{a} \rangle \models (X=x)
&& \text{(from \eqref{eq:modalext} and \eqref{eq:reducestate})} \\
\text{ iff } \quad & \vec{a}_{\d[\svc{Y}]}|_{X} = x
&& \text{(from \eqref{eq:lewisSem})}
\end{align*}

Now from \eqref{m.fixes}, this last line is bi-equivalent to the stipulation that $x$ is the $\ll^{X}_{\vec{a}}$-minimal element of $\{ x' \in \Rn(X) \mid \d[\svc{Z}\,] \fixes_{\vec{a}} (X= x')\}$, where $\vec{Z}$ denotes the (possibly empty) initial segment of $\dot{\leadsto}_{\vec{a}}$ (strictly) less than $X$ and $\vec{z}$ denote the values of this atom over $\vec{Z}$: $\vec{z} = \vec{a}_{\d[\svc{Y}]}|_{\vec{Z}}$. By our assumption $X \leadsto_{\vec{a}} W$, and so by Lemma \ref{lem:extendVars}, we have 
$$\{ x' \in \Rn(X) \mid \d[\svc{Z}\,] \fixes_{\vec{a}} (X= x')\} = \{ x' \in \Rn(X) \mid \d[\svc{Z}, W \leftarrow w] \fixes_{\vec{a}} (X= x')\}$$
and thus, by \eqref{m.fixes} that $\vec{a}_{\d[\svc{Y}]}, W \leftarrow w)|_{X} = x$. Running the chain of implications in reverse, we have that this is equivalent to 
$$
\pair{a} \models \d[\svc{Y},W \leftarrow w](X = x)
$$
as needed.
\end{proof}

This suffices to show our representation exists. Indeed, we can define
the needed utility function as  
\begin{equation}
\label{eq:utilitydef}
\mu(\vec{a}) = u(\langle \vec{a}, \vec{a}_{\vec{a}|_\U} \rangle).
\end{equation}
}

We can now complete the proof that $M^C$ represents the preference order. 
The utility function $\mu$ given by BH is defined on their states, which are
pairs of atoms. We define $\mu^C$ over individual atoms as
\begin{equation}
\label{eq:udef}
\mu^C(\vec{a}) = u(\langle \vec{a},\vec{a}_{\vec{a}|_\U} \rangle)
\end{equation}
if $\vec{a}|_\U \in \C^\dag$ (and define it arbitrarily otherwise). 
Unpacking this, $\vec{a}|_\U$ is the context of the atom $\vec{a}$, and so, 
$\vec{a}_{\vec{a}|_\U}$ is the unique non-null atom with the same context as $\vec{a}$.

\begin{align*}
\mu(\textsc{min}_{\langle \vec{a}_{\u}, \vec{a}_{\u} \rangle}(h_A( \vec{a}_{\u})))
=& \ \mu(\langle (\vec{a}_{\u})_{h_A(\vec{a}_{\u})}, \vec{a}_{\u} \rangle)
&& \text{(from \eqref{eq:mintoclosest})} \\
=& \  \mu^C( (\vec{a}_{\u})_{h_A(\vec{a}_{\u})} )
&& \text{(from \eqref{eq:udef})} \\
=& \ \mu^C( \beta_A^M(\u))
&& \text{(from \eqref{eq:betaviah})}
\end{align*}
Substituting $\mu^C$ for $\mu$ in \eqref{eq:nurep} thus delivers the
desired representation.  

\fullv{
To prove Theorem \ref{thm:uniq}, the uniqueness claim, let $(M,
\p,\mu)$ and $(M', \p',\mu')$ be two representations of $\s_\S$.  
For all $\u \in \C^\dag$ and all vectors $\vec{Y}$, $\vec{y}
\in \Rn(\vec{Y})$, 
endogenous variables $X \notin \vec{Y}$, and $x \in \Rn(X)$
\begin{equation}
\label{eq:uniq1}
(M_{\d[\svc{Y}]},\u) \models (X = x) 
\implies 
\d[\svc{Y}]
\fixes_{\vec{a}_{\u}} (X = x); 
\end{equation}
similarly,
\begin{equation}
\label{eq:uniq2}
(M'_{\d[\svc{Y}]},\u) \models (X = x) 
\implies 
\d[\svc{Y}]
\fixes_{\vec{a}_{\u}} (X = x). 
\end{equation}

Suppose that $\s_\S$ satisfies A\ref{ax:defp}$^*$. Then
there is a unique $x \in \Rn(X)$ that satisfies
the  
right-hand relations of \eqref{eq:uniq1} and \eqref{eq:uniq2}, so we have
$$
(M_{\d[\svc{Y}]},\u) \models (X = x)  
\iff 
(M'_{\d[\svc{Y}]},\u) \models (X = x). $$ 
It easily follows that $M$ and $M'$ agree on all formulas in
$\L^+(\S)$, so $M$ is identifiable.

For the converse, suppose that $M \neq M'$ for some $\u \in
\C^\dag$. In particular, 
this requires that there exists some vector $\vec{Y}$, $\vec{y} \in
\Rn(\vec{Y})$ and endogenous variable $X \notin \vec{Y}$, such that
$(M_{\d[\svc{Y}]},\u) \models (X = x)$ and $(M'_{\d[\svc{Y}]},\u)
\models (X = x')$ and $x \neq x'$. But then, by \eqref{eq:uniq1} and
\eqref{eq:uniq2}, $\s_\S$ violates A\ref{ax:defp}$^*$. 
}

\section{Discussion}

We have given a representation theorem in the spirit
of Savage \nciteyear{Savage} that helps us understand a decision
maker's (subjective) causal judgements: If a decision maker's
preferences among actions that involve interventions satisfy certain
axioms, then we can find a
causal model $M$, a probability over contexts in $M$, and a utility on
states in $M$ such that the decision maker prefers action $A$ to
$B$ iff $A$ has higher expected utility than $B$.
Moreover, we have shown if we add another axiom, then $M$ is unique.
Our approach builds on earlier work by BH, who proved an analogous
representation theorem using Lewis-style models.  Interestingly, other
than the Cancellation axiom (which is due to BEH), the axioms used by
BH are completely different from ours. For example, they do not define
an analogue of our $\leadsto_{\vec{a}}$ relation, which plays a
critical role in our definiteness and centeredness axioms.  Roughly
speaking, the BH axioms build on axioms for counterfactuals used by
Lewis, while ours build on axioms for causal models introduced by
Galles and Pearl \nciteyear{GallesPearl98} and
Halpern \nciteyear{Hal28}.  In future work, we hope to better
understand the connection between the axioms (e.g., to what extent
could we use the BH axioms as a basis for a representation theorem for
causal models), with the hope of understanding better the connection
between Lewis-style models and causal models.

Here we have considered only one-step decisions; that is, the
decision-maker performs an intervention, perhaps conditional on some
test.  It is clearly also of interest to consider sequential
decisions.  In practice, plans are composed of a sequence of
steps; later interventions might depend on earlier  interventions.
This leads us to consider a richer set of actions such as {\bf if $\phi_1$
then $A_1$ else $B_1$; if $\phi_2$ then $A_2$ else $B_2$}, where the
second action ({\bf if $\phi_2$ then $A_2$ else $B_2$}) is performed
after the first.  Getting a representation theorem for Lewis-style
models in the presence of sequential actions seems significantly more
complicated than in the case of one-step actions \cite{BH23}.
It would be of interest to see what is involved in getting such a
representation theorem using causal models.

Interestingly, we expect there to be significant differences between
causal models and Lewis-style once we have sequential decisions.
A decision maker who uses causal
models will behave in a
dynamically consistent way with respect to sequential interventions
involving primitive actions:
if the intervention $\d[\svc{Y},\svc{Z}]$ is preferred to the
intervention $\d[\svc{Y},\svc{Z'}]$, then $\d[\svc{Z}]$ must be
preferred to $\d[\svc{Z'}]$ in the model that results from taking
action $\d[\svc{Y}]$.
In other words, the decision maker's preferences regarding sequences
of primitive actions are invariant to
timing, or order, although this is not necessarily the
case for arbitrary interventions.  
This is because the interventions $\d[\svc{Y},\svc{Z}]$ and
$\d[\svc{Z},\svc{Y}]$ are identical.

By way of contrast, the order in which primitive interventions are
performed can have a significant impact in Lewis-style models.
In such models, the
closest-world operator is local; loosely speaking, the ``distance''
between two worlds depends on the world at which they are being
contemplated. Because of this, intervening to make $\phi$
true---moving to the closest $\phi$-world---also changes the relative
closeness of other worlds; thus, the closest $\phi \land \psi$-world
may bare no relation to the closest $\psi$-world to
the closest $\phi$-world.
This means that we lose some dynamic consistency in Lewis-style models.
\commentout{

We intend to explore, in future work, the issue of dynamic consistency
and order dependence in general Lewis-style models. In particular, we
hope to characterize dynamic consistency via conditions on a general
closeness operator. Clearly, the recursiveness condition
of \cite{Hal40}, by merit of establishing an equivalent structural
equations model, is sufficient. Thus, identifying the gap between this
and a necessary condition would establish a more full understanding of
the relationship between different causal semantics. } 
In future work, we hope to explore the issue of dynamic consistency and order
dependence in both causal models and Lewis-style models, both because
we believe that the issue of importance in its own right, and because
it will help elucidate the differences between these two approaches to
modeling causality.

\appendix

\shortv{
\section*{Ethical Statement}

There are no ethical issues.

\section*{Acknowledgments}
}


\bibliographystyle{named}
\bibliography{joe,z}

\shortv{

\clearpage

\section{Appendix}
\subsection{Proofs of Lemmas}

 \begin{lproof}{lem:minatom}
   Since $\U$ forms the initial segment of $\dot{\leadsto}$, it
is follows from \eqref{o.1} that \eqref{m.minY} must hold.
In more detail, if $\vec{a}_{\d[\svc{Y}]}|_{\U} \ne \vec{a}|_{\U}$, then
let $\vec{b}$ be such that $\vec{a}_{\d[\svc{Y}]}|_{\V} = \vec{b}|_{\V}$
and $\vec{b}_{\U} = \vec{a}|_{\U}$.  Since $\U$ forms the initial
segment of $\dot{\leadsto}$, it follows from \eqref{o.1}
that $\vec{b} <_{\vec{a}} \vec{a}_{\d[\svc{Y}]}|_{\U}$, yet
$\vec{b}$ satisfies  $\vec{Y} = \vec{y}$. This is a contradiction.

 To see that \eqref{m.fixes} holds, suppose by way of contradiction
 that it does not hold.  Then there exists a variable $X$ such that
the value of $X$ in  $\vec{a}_{\d[\svc{Y}]}$ is not the
$\ll^{X}_{\vec{a}}$-minimal element of
$m(\vec{a}, X)$ (note that $m(\vec{a}, X)$ is non-empty by Axiom
 A\ref{ax:def}).   
Let $\vec{b}$ be the atom that coincides with
$\vec{a}_{\d[\svc{Y}]}$ for all variables except $X$, and the value of $X$
in $\vec{b}$ is the
$\ll^{X}_{\vec{a}}$-minimal element of $m(\vec{a}, X)$.
There are two cases: 
 \begin{itemize}
 \item
   $\vec{a}|_{\vec{Z}_X} = \vec{a}_{\d[\svc{Y}]}|_{\vec{Z}_X}$:
   From Axiom A\ref{ax:cent}, it follows that $\d[\vec{Z}_X \leftarrow
    \vec{z}_X] \fixes_{\vec{a}} \vec{a}|_{X} $. By the definition of
       $\ll^{X}_{\vec{a}}$, the value of $X$ in $\vec{a}$ is the
       $\ll^{X}_{\vec{a}}$-minimal element of $m(\vec{a}, X)$. By
   assumption, the value of $X$ in  $\vec{a}_{\d[\svc{Y}]}|_{X}$ is not the
   $\ll^{X}_{\vec{a}}$-minimal element of $m(\vec{a}, X)$, while the
   value of $X$ in $\vec{b}$ is the  $\ll^{X}_{\vec{a}}$-minimal
   element of $m(\vec{a}, X)$.  It follows that $\vec{b}  <_{\vec{a}}
   \vec{a}_{\d[\svc{Y}]}$, a contradiction.
     \item $\vec{a}|_{\vec{Z}_X} \neq
       \vec{a}_{\d[\svc{Y}]}|_{\vec{Z}_X}$: Since $\vec{a}|_{\vec{Z}_X} \neq
       \vec{a}_{\d[\svc{Y}]}|_{\vec{Z}_X} = \vec{b}|_{\vec{Z}_X}$, we have
%
       $Y_{\vec{a}, \vec{a}, \vec{a}_{\d[\svc{Y}]}} =
       Y_{\vec{a}, \vec{a}, \vec{b}}$, so by
         \ref{o.2}, $\vec{b} <_{\vec{a}} \vec{a}_{\d[\svc{Y}]}$, and again
         we have a contradiction. 
\end{itemize}
  \end{lproof}

\begin{lemma}
  \label{lem:extendVars}
  For all atoms $\vec{a} \in \at$, (disjoint) vectors of endogenous variables
$\vec{Y}$ and  $\vec{Z}$, $\vec{y} \in \Rn(\vec{Y})$,
  $\vec{z} \in \Rn(\vec{Z})$, endogenous variables $X  
\notin \vec{Z}\cup\vec{Y}$ such that for all $Y\in \vec{Y}$, $X
\dot{\leadsto} Y$, and $x \in \Rn(X)$, we have
\begin{align*}
&\d[\svc{Z}\,] \fixes_{\vec{a}} (X= x) \quad \text{ iff } \\
&\d[\svc{Z},
    \svc{Y}] \fixes_{\vec{a}} (X= x). 
\end{align*}
\end{lemma}

\begin{lproof}{lem:extendVars}
  By assumption, for all $Y \in \vec{Y}$, we have $Y
 {\centernot{\leadsto}}_{\vec{a}}\, 
  X$. The lemma then follows from  a simple induction argument on the
variables of $\vec{Y}$, appealing to \eqref{eq:leadstocontra}  in each
step. 
\end{lproof}

\begin{lproof}{lem:sameIntervention}
  We prove this by induction on the variables in $\V$ with respect to  the order
   $\dot{\leadsto}$. Fix some $X \in \V\setminus \vec{Y}$, 
let $\vec{Z} \subseteq \V$ denote the (possibly empty) vector of
endogenous variables strictly $\dot{\leadsto}$-less than $X$, 
and let $\vec{z}$ denote the values of the variables $\vec{Z}$ in
$\vec{a}_{\d[\svc{Y}]}$.
%
To simplify notation, let $x$ denote the value of $X$ in
$\vec{a}_{\d[\svc{Y}]}$, let $\vec{Y}' \subseteq \vec{Y}$ consist of those variables
$\dot{\leadsto}$-greater than $X$, and let $\vec{y}^{\,
  \prime}$ be the restriction of $\vec{y}$ to $\vec{Y}'$.
We now show that if 
\begin{equation}
\label{eq:induct_hypo2}
\textbf{if } \vec{a}  \textbf{ then } \d [\svc{Y}] \\ \sim \\ \textbf{if } \vec{a}  \textbf{ then } \d [\svc{Z}, \vec{Y}' \leftarrow \vec{y}^{\, \prime}] 
\end{equation}
then 
\begin{gather}
\label{eq:induct_conseq}
\textbf{if } \vecc{a}  \textbf{ then } \d [\svc{Y}] \nonumber \\
\sim \\ 
\nonumber
\textbf{if } 
\vecc{a}  \textbf{ then } \d [\svc{Z}, \vec{Y}' \leftarrow \vec{y}^{\,
        \prime}, X \leftarrow x].  
\end{gather}
From \eqref{m.fixes}, it follows that $\d[\svc{Z}\,] \fixes_{\vec{a}} (X=
x)$. Applying Lemma \ref{lem:extendVars}, we obtain 
\begin{equation}
\label{eq:x_fixed_yprime}
\d[\svc{Z}, \vec{Y}' \leftarrow \vec{y}^{\, \prime}] \fixes_{\vec{a}} (X= x).
\end{equation}
Thus,
\begin{align*}
&\textbf{if } \vecc{a}  \textbf{ then } \d [\svc{Z}, \vec{Y}' \leftarrow \vec{y}^{\, \prime},X \leftarrow x] \\
\sim \quad &\textbf{if } \vecc{a}  \textbf{ then } \d [\svc{Z},
  \vec{Y}' \leftarrow \vec{y}^{\, \prime}] 
\\ 
\sim \quad 
&\textbf{if } \vecc{a}  \textbf{ then } \d [\vec{Y}
    \leftarrow \vec{y}], 
\end{align*}
where the first indifference comes from from \eqref{eq:x_fixed_yprime}
and the definition of $\fixes_{\vec{a}}$), 
and the second from \eqref{eq:induct_hypo2}.
  \end{lproof}

\begin{lproof}{lem:biequiv}
Let $Y_1, \ldots, Y_k$ be the ordering on endogenous variables
according to $\dot{\leadsto}$. 
We define $\F^C$ for the endogenous variables by induction $j$ so that
(\ref{eq:equivmodel1}) for all formulas $Y_j = y_j$.

For the base case, note that the value of the variable $Y_1$ does not
depend on the 
values of any other endogenous variable.  If is straightforward to
define $F_{Y_1}$ so that $F_{Y_1}(\vec{u}) = y_1$ iff
$a_{\vec{u}} \Rightarrow Y_1=y_1$.  It easily follows that 
$(M^C,\u) \models Y_1 = y_1$ iff 
$(M,\langle  \vec{a}_{\u}, \vec{a}_{\u} \rangle) \models Y_1 = y_1$.
Moreover, if $Y_1 \notin \vec{Y}$, then 
$(M^C_{\d[\svc{Y}]},\u) \models Y_1 = y_1$ iff
$(M^C,\u) \models Y_1 = y_1$.  
By Lemma~\ref{lem:extendVars} (taking $\vec{Z} = \emptyset$ and $X = Y_1$)
it follows that $(M,\textsc{min}_{\langle \vec{a}_{\u}, \vec{a}_{\u}
  \rangle}(\d[\svc{Y}])) \models Y_1 = y_1$ iff $(M,\langle \vec{a}_{\u},
\vec{a}_{\u}) \models Y_1 = y_1$.  On the other hand, if $Y_1 \in
\vec{Y}$, then $(M^C_{\d[\svc{Y}]},\u) \models Y_1 = y_1$ iff
$\vec{y}\mid_{Y_1} = y_1$ and similarly, $(M,\textsc{min}_{\langle \vec{a}_{\u}, \vec{a}_{\u}
  \rangle}(\d[\svc{Y}])) \models Y_1 = y_1$ iff  $\vec{y}\mid_{Y_1} = y_1$.
Thus, we get the desired result in the base case.

For the inductive case, suppose that we have proved the result for
$Y_1, \ldots, Y_j$.  We can define $F_{Y_{j+1}}$ to be function of the
  exogenous variables and $Y_1, \ldots, Y_j$ such that for all $k \in
  \{1,\ldots, j\}$, and $y_i \in \Rn(Y_i)$ for $i = 1,\ldots, j$, we
  have that $F_{Y_{j+1}}(\vec{u}, y_1, \ldots, y_j) = y_{k+1}$ iff
$(M,\textsc{min}_{\langle \vec{a}_{\u}, \vec{a}_{\u}
  \rangle}(\d[Y_1 \gets y_1, \ldots, Y_k \gets y_k]) \models Y_{k+1}
= y_{k+1}$.  Now for arbitrary $\vec{Y}$, if $Y_{j+1} \notin \vec{Y}$,
let $\vec{Y}' = \vec{Y} \cap
\{Y_1, \ldots, Y_j\}$.  Let $\vec{y}' = \vec{y}|_{\vec{Y}'}$.  Writing
$\vec{Y}''$ for $\{Y_1, \ldots, Y_j\}$, let $\vec{y}''$ be such that
$(M^C_{\d[\vec{Y}' \gets \vec{y}']},\u) \models \vec{Y}'' =
\vec{y}''$.  Then it follows from the induction hypothesis that 
 $(M,\textsc{min}_{\langle \vec{a}_{\u}, \vec{a}_{\u}
  \rangle}(\d[\vec{Y}' \gets \vec{y}'])) \models \vec{Y}'' =
\vec{y}''$.
Moreover, it easily follows that $(M^C_{\d[\svc{Y}]},\u) \models
Y_{j+1} = y_{j+1}$ iff $(M^C_{\d[\vec{Y}'' \gets\vec{y}'']},\u) \models
Y_{j+1} = y_{j+1}$ and 
 $(M,\textsc{min}_{\langle \vec{a}_{\u}, \vec{a}_{\u}
  \rangle}(\d[\vec{Y} \gets \vec{y}])) \models Y_{j+1} =
y_{j+1}$ iff 
 $(M,\textsc{min}_{\langle \vec{a}_{\u}, \vec{a}_{\u}
  \rangle}(\d[\vec{Y}'' \gets \vec{y}''])) \models Y_{j+1} =
y_{j+1}$.  Since, by the definition of $F_{Y_{j+1}}$,
$(M^C_{\d[\vec{Y}'' \gets\vec{y}'']},\u) \models Y_{j+1} = y_{j+1}$ 
iff
 $(M,\textsc{min}_{\langle \vec{a}_{\u}, \vec{a}_{\u}
  \rangle}(\d[\vec{Y}'' \gets \vec{y}''])) \models Y_{j+1} =
y_{j+1}$, the result easily follows.  The argument if $Y_{j+1} \in
\vec{Y}$ is the same as in the base 
case.  

The result for arbitrary
formulas $\phi \in \L(\S)$ is immediate (since conjunction,
disjunction, and negation work the same way in both $M$ and $M^C$).
\end{lproof}

\subsection{Uniqueness}

Consider the following strengthening of A\ref{ax:def}:
\addtocounter{ax}{-3}
\begin{axprime}[Strong Definiteness]
\label{ax:defp}
For each non-null atom $\vec{a}$, vector of $\vec{Y}$ of endogenous variables,
$\vec{y} \in \Rn(\vec{Y})$,
and endogenous variable $X \notin \vec{Y}$, there
exists a \emph{unique} $x \in \Rn(X)$ such that $\d[\svc{Y}] \fixes_{\vec{a}} (X = x)$. 
\end{axprime}

Strengthening A\ref{ax:def} to A\ref{ax:defp}$^*$, is both necessary
and sufficient to avoid multiple representations.

\begin{definition}
  A subjective causal utility representation of a preference relation of
$(M',\p',\mu')$, non-null contexts $\u$,  and formulas $\phi \in \L^+(\S)$,
$$ (M,\u) \models \phi \qquad \text{ if and only if} \qquad (M',\u)
  \models \phi.$$ 
\end{definition}

\begin{theorem}
\label{thm:uniq}
A subjective causal utility representation $(M,\p,\mu)$ of
$\s_\S$ is \emph{identified} if and only if $\s_\S$ satisfies 
Axiom A\ref{ax:defp}$^*$.
\end{theorem}

\begin{proof}
To prove Theorem \ref{thm:uniq}, the uniqueness claim, let $(M,
\p,\mu)$ and $(M', \p',\mu')$ be two representations of $\s_\S$.  
For all $\u \in \C^\dag$ and all vectors $\vec{Y}$, $\vec{y}
\in \Rn(\vec{Y})$, 
endogenous variables $X \notin \vec{Y}$, and $x \in \Rn(X)$
\begin{equation}
\label{eq:uniq1}
(M_{\d[\svc{Y}]},\u) \models (X = x) 
\implies 
\d[\svc{Y}]
\fixes_{\vec{a}_{\u}} (X = x); 
\end{equation}
similarly,
\begin{equation}
\label{eq:uniq2}
(M'_{\d[\svc{Y}]},\u) \models (X = x) 
\implies 
\d[\svc{Y}]
\fixes_{\vec{a}_{\u}} (X = x). 
\end{equation}

Suppose that $\s_\S$ satisfies A\ref{ax:defp}$^*$. Then
there is a unique $x \in \Rn(X)$ that satisfies
the  
right-hand relations of \eqref{eq:uniq1} and \eqref{eq:uniq2}, so we have
$$
(M_{\d[\svc{Y}]},\u) \models (X = x)  
\iff 
(M'_{\d[\svc{Y}]},\u) \models (X = x). $$ 
It easily follows that $M$ and $M'$ agree on all formulas in
$\L^+(\S)$, so $M$ is identifiable.

For the converse, suppose that $M \neq M'$ for some $\u \in
\C^\dag$. In particular, 
this requires that there exists some vector $\vec{Y}$, $\vec{y} \in
\Rn(\vec{Y})$ and endogenous variable $X \notin \vec{Y}$, such that
$(M_{\d[\svc{Y}]},\u) \models (X = x)$ and $(M'_{\d[\svc{Y}]},\u)
\models (X = x')$ and $x \neq x'$. But then, by \eqref{eq:uniq1} and
\eqref{eq:uniq2}, $\s_\S$ violates A\ref{ax:defp}$^*$. 
\end{proof}
}
\end{document}